\documentclass[a4paper,10pt,reqno
]{article}

\usepackage[top=2.5cm, bottom=2.5cm, left=2.5cm, right=2.5cm]{geometry}

\usepackage{nomencl}
\usepackage{booktabs}
\usepackage{amssymb,amsmath,amsthm,ascmac,array,bm}
\usepackage[dvipdfmx]{graphicx}
\usepackage{xcolor}
\usepackage{mathtools}
\usepackage{authblk}
\usepackage{multirow}
\usepackage{lscape} 

\usepackage[
colorlinks=true,bookmarks=true,
bookmarksnumbered=true,bookmarkstype=toc,linktocpage=true
]{hyperref}
\usepackage{cleveref}

\newtheorem{thm}{Theorem}[section]

\newtheorem{lem}[thm]{Lemma}
\newtheorem{prop}[thm]{Proposition}

\newtheorem{rem}[thm]{Remark}

\numberwithin{equation}{section}
\allowdisplaybreaks

  \newcommand{\gam}{\gamma}

\def\var{{\rm var}}
\def\cov{{\rm cov}}

\def\MSE{{\rm MSE}}
\def\MSPE{{\rm MSPE}}

\def\sumj{\sum_{j=1}^{n}}

\def\lambp{\lambda_i^{(BP)}}
\def\lamd{\lambda_i^{(d)}}

\def\lamp{\lambda_i^{(P)}}
\def\lamppr{\lambda_i^{(P),(PR)}}
\def\lampyl{\lambda_i^{(P),(YL)}}
\def\lamb{\lambda_i^{(B)}}
\def\lambpr{\lambda_i^{(B),(PR)}}
\def\lambyl{\lambda_i^{(B),(YL)}}
\def\lamroot{\lambda_{0i}}
\def\lamrootpr{\lambda_{0i}^{(PR)}}
\def\lamrootyl{\lambda_{0i}^{(YL)}}
\def\quarter{\frac{1}{4}}
\def\xib{(X_i\beta)}
\def\hxib{(\hat X_i\beta)}
\def\xhatb{\hat X_i\beta}

\def\exp{\mathbb{E}}
\def\prob{\mathbb{P}}
\def\zx{(Z_i - \hat X_i\beta)}
\def\fti{\frac{1}{t_i}}
\def\fbt{\frac{b_i}{t_i}}

\graphicspath{{plots/}}

\title{Small Area Estimation under Square Root Transformed Fay-Herriot model with Functional Measurement Error in Covariates}

\author[1]{Ka Long Keith Ho}
\author[2]{Masayo Y. Hirose\footnote{Corresponding Author, E-mail address: masayo@imi.kyushu-u.ac.jp}}
\author[3]{Malay Ghosh}

\affil[1]{Joint Graduate School of Mathematics for Innovation, Kyushu University}
\affil[2]{Institute of Mathematics for
Industry, Kyushu University}
\affil[3]{Department of Statistics, University of Florida}
\newcommand{\Addresses}{{
  \bigskip
  \bigskip
 Ka Long Keith Ho, \textsc{Joint Graduate School of Mathematics for Innovation, Kyushu University, 744 Motooka Nishi-ku Fukuoka 819-0395, Japan}\par\nopagebreak
  \textit{E-mail address}: \texttt{ho.kalongkeith.224@s.kyushu-u.ac.jp}

  \medskip

Masayo Y. Hirose, \textsc{Institute of Mathematics for
Industry, Kyushu University, 744
Motooka, Nishi-ku, Fukuoka, Japan}\par\nopagebreak
  \textit{E-mail address}, \texttt{masayo@imi.kyushu-u.ac.jp}

  \medskip

  Malay Ghosh, \textsc{Department of Statistics, University
of Florida, 223 Griffin-Floyd Hall,
Gainesville, FL, USA}\par\nopagebreak
  \textit{E-mail address}, \texttt{ghoshm@ufl.edu}
}}
\date{\today}

\begin{document}
\setlength{\baselineskip}{4.5mm}

\maketitle

\begin{abstract}
We consider a small area estimation model under square-root transformation in the presence of functional measurement error. When measurement error is present, the Bayes predictor can no longer be used as it depends on the covariates even if parameters are known. Therefore suitable replacements are called for, and we propose a predictor that only depends on observed responses and data obtained from a large secondary survey. Moreover, some estimating methods of unknown parameters are considered. In the simulations section, We evaluate the performance using the mean squared prediction error ($\MSPE$) and discuss several scenarios in terms of the number of areas and the sample size in a large secondary survey.
\end{abstract}

\section{Introduction}\label{intro}

The demand for reliable estimation has been growing, especially in small geographic areas and sub-populations and it is now widely recognized that small area estimation needs to be model-based. The classic paper of Fay and Herriot (1979) \cite{FH1979} introduced a normal mixed effects model for small area estimation with area-level random effects. Since then, there has been tremendous growth in the small area literature including analysis of binary and count data. We refer to Molina and Rao (2015) \cite{rao2015} for details on small-area estimation.

Quite often one encounters skewed data where normality can be justified only after a suitable transformation. A commonly used transformation is the logarithmic transformation, for example in handling income data in Slud and Maiti (2006) \cite{Slud2006}, and Ghosh et al. (2015) \cite{Kubokawa2015} used back transformation in the small area context.

One important issue that has surfaced very recently is measurement error in the covariates which is typical in some of the small area problems. This issue was first addressed in Ybarra and Lohr (2008) \cite{YL2008}. This has been followed up in several other articles (Torabi (2012)\cite{Torabi2012}, Datta et al., (2018) \cite{DATTA2018}, Mosaferi et al. (2023) \cite{Mosaferi} etc.), with a multivariate extension given in Arima et al. (2017) \cite{Arima2017}.

Variance stabilizing variable transformation has long been in the literature starting with Bartlett (1936) \cite{SqaureRoot1936} and Anscombe (1948) \cite{Anscombe1948}. Other than closer approximations to normality, the virtue of such transformations is to produce an approximate variance free from unknown parameters and is particularly useful for binary and count data where variances are functions of the means. The arc sin transformation for binary data and square root transformation for the count data achieve this. Ghosh et al. (2022) \cite{Hirose2022}, and Hirose et al. (2023) \cite{arcsin2023} used such transformations in the small area context for binary and count data respectively. 

The objective of this article is to use the square root transformation in small area problems where covariates are measured with error. Earlier, Mosaferi et al. (2023) \cite{Mosaferi} and Mosaferi et al. (submitted for publication) have addressed such problems in the context of logarithmically transformed data. 

The paper is arranged as follows. We formulate the problem in \Cref{secmodel} and consider the predictors assuming known model parameters in \Cref{secmain}, where we also give expressions for the bias and $\MSPE$ of the considered class of predictors. In \Cref{secEST} we discuss estimation schemes for model parameters and give the empirical versions of the predictors presented in \Cref{secmain}. Simulation results are given in \Cref{secSIM} and \Cref{secDIS} contains some formal remarks. Proofs are given in the Appendix. 

\section{Model Setup}\label{secmodel}
We consider a small-area functional measurement error model with $m$ areas of interest. In area $i = 1,...,m$, we observe a count $Y_i$ that is supposed to be independently Poisson distributed with an unknown mean $\lambda_i$, so that \[Y_i \sim Poisson(\lambda_i) \text{  for  } i = 1,...,m.\] Under the square-root transformation (Bartlett (1936) \cite{SqaureRoot1936}) for the Poisson distribution, we obtain $\var[\sqrt{Y_i}] = \var[Z_i]\approx \quarter$. We denote $\theta_i = \sqrt{\lambda_i}$ and assume 
\[\theta_i = X_i\beta + v_i \sim N(X_i\beta,a)\]
for unobserved $p$-dimensional covariates $X_i$ and regression parameter $\beta \in \mathbb{R}^p$. The random effect $v_i$ is supposed to be normally distributed with $ v_i \sim N(0,a)$ under unknown variance $a \in \mathbb{R}^+$. In this paper, we consider
\[Z_i = \theta_i + e_i,\]
for sampling error $e_i \sim N(0,\quarter).$ From other surveys, we consider $t_i$ surrogate covariates $\hat X_{ij}$ of $X_i$ in area $i$,
where $\mathbb{E}[\hat X_{ij}] = X_i$ and $\cov[\hat X_{ij}] = \Sigma_i$ are i.i.d. for $j = 1,...,t_i$, where for $i = 1,...,m$, $\Sigma_i$ is a known positive definite matrix. We assume $\inf_i t_i$ is large, whose purpose is to consider a diminishing effect of measurement error when sufficient auxiliary data is obtained. Therefore we appeal to the Central Limit Theorem and assume
\[\hat X_i = \frac{1}{t_i}\sumj \hat X_{ij} = X_i + \Delta_i \sim N\left(X_i,\frac{1}{t_i}\Sigma_i\right), \] where we denote $\Delta_i = \frac{1}{t_i}\sumj \Delta_{ij}$. Also, we note that
\[(\hat X_i - X_i)\beta \sim N\left(0,\frac{1}{t_i}\beta^T\Sigma_i\beta\right),\]
as we denote $b_i = \beta^T\Sigma_i\beta$ for convenience. We assume the errors $(e_i,v_i,\Delta_{ij})$ are independent across all $i$ and $j$ and are mutually independent of each other, and thus $Z_i$ and $\hat X_i$ are independent for all $i = 1,...,m$. This model is similar to the Fay-Herriot Model (Fay and Herriot (1979) \cite{FH1979}) under measurement error studied by Ybarra and Lohr (2008) \cite{YL2008}, and here we consider an extra layer of complexity with the square-root transformation. Whilst the normality assumption on $\Delta_i$ is usually not compulsory in measurement error models, we incorporate this assumption as it is natural in the view of Bayes predictors.

Ybarra and Lohr (2008) \cite{YL2008} considered the same problem without square-root estimation. They considered the predictors for $y_i$ (analogous to $\theta_i$ in this paper) of the form $w_i Y_i + (1 - w_i) \hat X_i \beta$, where $Y_i$ are the observed responses and $\hat X_i$ are the observed covariates under functional measurement error. When the parameters $a$ and $\beta$ are known, they proved that the optimal choice is $w_i = 1 - \frac{\psi_i}{a_i + b_i + \psi_i}$, where $\psi_i$ is the variance of the sampling error, which is analogous to $1/4$ in our setting. 

In this paper, we consider predictors of the form $(w_i Y_i + (1 - w_i) \hat X_i \beta)^2 + C$, where C is a constant. This consideration is similar to the Bayes predictor given in Ghosh et al. (2022) \cite{Hirose2022}, which will be elaborated in \Cref{secmain}). We justify the choice $w_i = 1 - \frac{1/4}{a_i + b_i + 1/4}$ of Ybarra and Lohr (2008) \cite{YL2008} under the currently considered model. We will give the explicit expression of the mean squared prediction error ($\MSPE$ hereon) of this class of predictors in \Cref{thmMSE} and use it for evaluation of performances. 

We are further motivated to consider the scenario when an abundant amount of data from secondary surveys are available, even if some areas are having small sample sizes $n_i$ in the original data set, which propagates us to study results under large $\inf_i t_i$. For example, Datta et al. (2018) \cite{DATTA2018} analyzed BMI data from the 2013-14 US NHANES with the 2014 US NHIS as auxiliary information with $m = 50$ small domains. In their study, $t_i$ ranged from 77 to 4866, collected from 35928 individuals in total, which is considerably larger than their domain sample sizes of 31 to 479 from 5588 individuals. The consideration of large $t_i$ is not present in similar literature (Ybarra and Lohr (2008) \cite{YL2008}, Mosaferi et al.(2022) \cite{Mosaferi}), but is going to be useful in recovering the measurement error-free Bayes predictor and provides an extra layer of justification for our proposed predictor, which will be elaborated in later sections. Besides, Mosaferi et al. (2023) \cite{Mosaferi} suggested the predictor under log transformed Fay-Herriot model with measurement error. However, it seems not to correct bias in the presence of measurement error. So we obtained the bias and corrected how to propose a predictor in this study, while taking into account of MSPE.

\section{Predictors in the Presence of Measurement Error}\label{secmain}
\subsection{Bayes Predictor and Related Class of Predictors}\label{secmainintro}
We introduce and proposed several predictors for $\lambda_i$ in this section when the model parameters $(a,\beta)$ are known. When there is no measurement error, the best predictor of $\lambda_i$ that has the minimum $\MSPE$ in this setting is the Bayes predictor given by 
\begin{equation}\label{BP}
\lambp = \exp[\lambda_i|Z_i]=\{(1-B)Z_i+B\xib\}^2+\frac{1-B}{4} \text{  for  } i = 1,...,m
\end{equation}
for $B = \frac{1/4}{1/4+a}$, whose empirical version was studied recently (Ghosh et al. (2022) \cite{Hirose2022}). 
In the presence of measurement error, the Bayes predictor can no longer be used as $X_i\beta$ is unknown. Thus we turn to $\xhatb$ as an alternative, and as a consequence, we must account for the effect of measurement error represented by $\Delta_i$. Inspired by Ybarra and Lohr \cite{YL2008} and the Bayes predictor given in Ghosh et al. (2022) \cite{Hirose2022}, we consider predictors of the form 
\[\hat\lambda_i(w_i;C_i) = \{w_iZ_i + (1-w_i)\xhatb\}^2+C_i,\]
with a weight $w_i \in [0,1]$ and $C_i$ a constant value free from $Z_i$ and $X_i\beta$. The $\MSPE$ of any predictor is given by the sum of its variance and its bias squared by the bias-variance decomposition. Therefore any biased predictor can always be improved by adjusting for bias when it can be calculated. We give an explicit expression of the bias in the following Lemma so that for any fixed $w_i$ the optimal choice for $C_i$ is \[C_i = -\exp\left[\{w_iZ_i + (1-w_i)\xhatb\}^2 - \lambda_i\right] \text{  for  } i = 1,...,m.\]

\begin{lem}\label{bias}
For each $i = 1,...,m$, \[\exp\left[\{w_i Z_i + (1-w_i)\xhatb\}^2 - \lambda_i\right] = w_i^2(a+1/4) + \fti (1-w_i)^2b_i-a.\]
\end{lem}

\begin{proof}
By independence of $Z_i$ and $\xhatb$, \[w_iZ_i + (1-w_i)\xhatb \sim N\left(X_i\beta, w_i^2(a+\quarter)+\fti (1-w_i)^2b_i\right).\] 
Taking expectations, we have 
\begin{align*}
    \mathbb{E} \left[ \{w_iZ_i + (1-w_i)\xhatb \}^2 - \lambda_i \right] &= \xib^2 + w_i^2(a+\quarter) + \fti (1-w_i)^2b_i - \xib^2 - a \\
    &= w_i^2(a+\quarter)+ \fti (1-w_i)^2 b_i - a.
\end{align*}
\end{proof}
Henceforth we will assume this value of $C_i$ for all predictors unless otherwise specified. For each value of $w_i$, we may also write down the analytical $\MSPE$ using \Cref{bias}:
\begin{thm}\label{thmMSE}
The $\MSPE$ of $\hat\lambda_i(w_i)$ is given by
\begin{equation}\label{MSEwi}
    \begin{split}
    \MSPE[\hat\lambda_i(w_i)] &= \xib^2\left[(4a+\frac{4b_i}{t_i}+1)w_i^2-(8a+\frac{8b_i}{t_i})w_i+4a+4
    \fbt \right]\\
     &+ \bigg[3 \bigg\{w_i^2(a+\quarter)+(1-w_i)^2\fbt\bigg\}^2 - \bigg\{w_i^2(a+\quarter)+(1-w_i)^2\fbt - a\bigg\}^2\\
     &+ 3a^2 -6a^2w_i^2-\frac{aw_i^2}{2}-2a\fbt(1-w_i)^2 \bigg].
    \end{split}
\end{equation}
\end{thm}
\begin{proof}
    See Appendix.
\end{proof}
\begin{rem}
Note that $\MSPE$ of $\hat\lambda_i$ depends on $X_i \beta$ but $X_i$ is unobserved. Therefore any estimators of the $\MSPE$ must incorporate an estimate for $X_i \beta$ as well.
\end{rem}
\subsection{Direct Predictor}\label{Direct}
Among all choices of $w_i$, one may choose to ignore any auxiliary information, which is equivalent to choosing $w_i = 1$, giving the direct predictor that is only based on responses within each area: \[\lamd = \hat\lambda_i(1) = Z_i^2 - 1/4.\]
While it is not the goal to advocate users to ignore strength that can be borrowed from other surveys, we note that the direct predictor is independent of $\xhatb$. Therefore, it does not suffer from any inaccuracies in estimating $\beta$ and $a$. The direct predictor also serves as a reference and provides a fall-back option if other issues persist, for instance, if parameters $a$ and $\beta$ cannot be reliably estimated (if $m$ is too small). One may use \Cref{MSEwi} to find that $\MSE[\lamd] = \xib^2 + a + \frac{1}{8}$.

\begin{rem}
    As we will show in \Cref{proopt}, the direct estimator can produce negative estimators. Despite not taking our specified form, we may use $\max (\lamd,0)$ as a straightforward improvement to $\lamd$.
\end{rem}

\subsection{Optimal Weight}\label{opt}
In the non-trivial case where $w_i \in (0,1)$, we shall search for candidates $w_i$ to achieve lower $\MSPE$. Ybarra and Lohr (2008) \cite{YL2008} showed that the optimal choice for the predictor $\theta_i$ is $w_i = 1 - \gamma_i$, in un-transformed case, where \[\gamma_i = \dfrac{1/4}{1/4+a+\fbt}.\]Yet the $\MSPE$ of their estimator is independent of $X_i\beta$, which led to the straightforward derivation. We showed that this is no longer the case, and in general, $w_i = 1 - \gamma_i$ does not minimize \eqref{MSEwi}. Treating \eqref{MSEwi} as a quartic polynomial in $w_i$, we can employ standard calculus methods to see that the minimizer of \eqref{MSEwi} must be a root of its derivative (up to a constant multiple) given by
\begin{equation}\label{dMSE}
   g\left(w_i;a,\fbt,\xib^2\right) = w_i^3-3\gamma_{bi}w_i^2+(2\gamma_{bi}^2-\gamma_{ai}^2+\gamma_{bi}+\gamma_{xi})w_i-(\gamma_{bi}^2+(\gamma_{ai}+\gamma_{bi})\gamma_{xi}),
\end{equation}
where \[\gamma_{ai} = \frac{a}{a+1/4+\fbt} \hspace{0.5cm} \gamma_{bi} = \frac{\fbt}{a+1/4+\fbt} \hspace{0.5cm}\text{and} \hspace{0.5cm} \gamma_{xi} = \frac{\xib^2}{a+1/4+\fbt}.\]
We may directly substitute $1 - \gam_i$ into $g$ to see that it is not equal to 0, hence it cannot be a minimizer of the $\MSPE$.

We can find explicit solutions for its roots through standard cubic solving techniques (Abramowitz and Stegun, page 17, \cite{cubic}) and denote the minimizer of \eqref{MSEwi} to be $w_{0i}$, but note that in practice cannot be explicitly found as $a$, $\beta$, and $X_i$'s are unobserved. We will denote $\lambda_{0i} = \hat\lambda_i(w_{0i})$ as the predictor with the optimal value of $w_i$.

\subsection{Proposed Weight}\label{pro}
While one might consider substituting $1 - B$ into \Cref{BP} in a similar model, we note that by \Cref{bias} the estimator is biased, nor does it have minimal $\MSPE$ by \Cref{thmMSE} even if it is bias-corrected. For comparison purposes later, we will denote $\lambda_i^{(B)} = \{(1- B) Z_i + B \xhatb\}^2 + \frac{1-B}{4}$ but note that this is not equal to $\hat\lambda_i (1 - B)$. 

On the other hand, despite not being the analytical minimum, we present other incentives that exist to choose $w_i = 1 - \gam_i$ on top of the predictor's simplicity and user-friendliness. As such we denote 
\begin{equation}\label{lamp}
    \lamp = \hat\lambda_i(1 - \gam_i) = \{(1-\gamma_i)Z_i+\gamma_i\hxib\}^2 - \gamma_i\left(\fbt-a\right),
\end{equation}
and we present a couple of intuitive justifications for this choice.

First of all, noting that $\theta_i = Z_i - e_i$, we may attempt to replace $e_i$ by its posterior expectation $\exp[e_i|Z_i-\xhatb]$ given the residual, which provides the predictor for $\theta_i$ among a linear class under functional measurement error under the normality assumption (Ybarra and Lohr (2008) \cite{YL2008}). Note that 
\[Z_i - \hat X_i \beta = r_i + e_i,\] and $(e_i,e_i + r_i)$ are jointly normal, which gives \[e_i|\zx \sim N\left(\gamma_i\zx, \frac{(1-\gamma_i)}{4}\right).\]
Furthermore, replacing $e_i$ by $\exp[e_i|\zx] = \gam_i\zx$ is equivalent to choosing $w_i = 1-\gamma_i$, and by \Cref{bias} we obtain
\begin{align*}
    \exp[\{(1-\gamma_i)Z_i+\gamma_i\hxib\}^2 - \lambda_i] &= \left(1-\gamma_i\right)^2\left(a+\quarter\right)+\gamma_i^2\fbt-a\\
    &= \gamma_i^2\left(a+\quarter+\fbt\right)-2\gamma_i\left(a+\quarter\right)+\left(a+\quarter\right)-a\\
    &=\quarter \gamma_i - 2\gamma_i\left(a+\quarter\right)+\gamma_i\left(a+\quarter+\fbt\right)\\
    &=\gamma_i\left(\quarter-2a-\frac{2}{4}+a+\quarter+\fbt\right)\\
    &=\gamma_i\left(\fbt-a\right),
\end{align*}
noting that $\quarter = \gamma_i\left(a+\quarter+\fbt\right)$. Combining gives
    \[\lamp = \{(1-\gamma_i)Z_i+\gamma_i\hxib\}^2 - \gamma_i\left(\fbt-a\right),\]
as suggested in \Cref{lamp}.

Another way to derive the same predictor is by considering the Maximum Likelihood Estimator (MLE) or the Uniformly Minimum-Variance Unbiased Estimator (UMVUE) for the unknown parameter $X_i\beta$ using both observations $\xhatb$ and $Z_i$. 
\begin{prop}
    The MLE and UMVUE of $X_i\beta$ are given by
    \begin{equation}\label{MLE}
        T_i = \gam_b Z_i + (1-\gam_b)\xhatb,
    \end{equation}
    where $\gam_b = \dfrac{\fbt}{a+1/4+\fbt}$.
\end{prop}
The proof for the MLE statement can be obtained directly by considering the likelihood equations. The proof for the UMVUE statement can be obtained using the Lehmann Scheffe Theorem, using the minimal sufficiency property of the exponential family, to which the multivariate normal belongs.
We return to \Cref{BP} and plug in $T_i$ defined in \Cref{MLE} for $\xib$. Since the predictor is determined solely by $w_i$, and
\begin{equation}\label{MLEplug}
    \begin{split}
    (1-B)Z_i+BT_i &= (1-B)Z_i+B(\gam_b Z_i + (1-\gam_b)\xhatb)\\
    &= (1-B + B\gam_b)Z_i + B(1-\gam_b)\xhatb\\
    &=\left(1 - \dfrac{1/4}{a+1/4}+\dfrac{1/4}{a+1/4}\dfrac{\fbt}{a+1/4+\fbt}\right)Z_i + \dfrac{1/4}{a+1/4}\dfrac{a+1/4}{a+1/4+\fbt}\xhatb\\
    &= (1-\gamma_i)Z_i + \gamma_i \hxib,
    \end{split}
\end{equation}
we obtain $\lamp$ once again, as in \Cref{lamp}. 

\subsection{The Proposed weight as an alternative of the Optimal Weight}\label{proopt}
We observe that $\gam_i - B \approx 0$ when $t_i$ is large for all areas $i = 1,...,m$, which allows us to recover the Bayes estimator when the effect of measurement error diminishes. We can consider the case of large $\inf_i t_i$ when observations from a secondary survey are abundant, giving the same result across all areas. As $w_i = 1 - B$ is also the minimizer of the $\MSPE$ without measurement error, we also observe that $w_{0i} - (1 - B) \approx 0$ for large $t_i$. In practice, $w_{0i}$ and $1 - \gam_i$ are often close to each other for large $t_i$, and due to the simple form of $1- \gamma_i$, we suggest it as a desirable choice. Write $g_{1i}(w_i)$ and $g_{2i}(w_i)$, so that in \Cref{MSEwi} we have
\begin{align*}
    \MSPE[\hat\lambda_i(w_i)] 
     = \xib^2 g_{1i}(w_i) + g_{2i}(w_i).
\end{align*}
Note that $1 - \gam_i$ minimizes $g_{1i}$, and let $w_{2i}$ denote the minimizer or $g_{2i}$ restricted to the closed interval $[0,1]$. Then 
\begin{align*}
\xib^2\{g_{1i}(1 - \gam_i) - g_{1i}(w_{0i})\} + \{g_{2i}(1 - \gam_i) - g_{2i}(w_{2i})\} > 0,
\end{align*}
which rearranges to 
\[g_{1i}(w_{0i}) - g_{1i}(1 - \gam_i)< \frac{g_{2i}(1 - \gam_i) - g_{2i}(w_{2i})}{\xib^2}.\]
Since $1 - \gam_i$ minimizes $f_{1i}$, we rewrite
\begin{align*}
    g_{1i}(w_{i}) = \left(4a +4\fbt+1 \right)\left\{w_i - (1 - \gam_i)\right\}^2 + \left(4a +4\fbt \right)\gam_i,
\end{align*}
and obtain
\begin{align*}
    g_{1i}(w_{0i}) - g_{1i}(1-\gam_i) = \left(4a +4\fbt+1 \right)\left\{ w_{0i} - (1-\gam_i)\right\}^2.
\end{align*}
Noting that $g_{2i}(1 - \gam_i) - g_{2i}(w_{2i}) > 0$, we obtain the bound
\begin{equation}\label{rootbound}
    |w_{0i} - (1 - \gam_i)| < \sqrt{\frac{g_{2i}(1 - \gam_i) - g_{2i}(w_{2i})}{\xib^2(4a+4\fbt+1)}}
\end{equation}
whose proficiency scales with the size of $\xib^2$.

\subsection{Range Preservation}\label{RangePres}
Note that the predictor $\hat\lambda_i(w_i)$ admits negative values with positive probability if the bias given in \Cref{bias} is positive. For instance, in the case of our predictor $\lamp$, if $\fbt > a$ then this probability is positive, which propagates us to define $\hat\lambda_{i+}(w_i) = max(0,\hat\lambda_i(w_i))$ as a straightforward improvement to $\hat\lambda_i$, even though it is not in the aforementioned form. 
By the independence of $Z_i$ and $\hat X_i\beta$, \[w_i Z_i + (1 - w_i) \hat X_i\beta \sim N\left(X_i\beta, w_i^2(a+1/4)+(1-w_i)^2 \fbt \right).\]
When the bias given by \Cref{bias} $C_i(w_i) > 0$, the probability of obtaining a negative predictor is
\begin{equation}\label{probeq}
    \begin{split}
    &\prob\left[-\sqrt{C_i} < w_i Z_i + (1 - w_i) \hat X_i\beta < \sqrt{C_i}\right]\\
    =&  \prob\left[Z \in \left(\frac{- X_i\beta -\sqrt{C_i} }{\sqrt{w_i^2(a+1/4)+(1-w_i)^2 \fbt}}, \frac{- X_i\beta + \sqrt{C_i} }{\sqrt{w_i^2(a+1/4)+(1-w_i)^2 \fbt}}\right)\right],
    \end{split}
\end{equation}
where $Z$ denotes a standard normal random variable independent of $Z_i$ and $X_i \beta$, and is not to be confused with $Z_i$. Note that this probability is small provided that $(X_i\beta)^2$ is moderately large and that the probability also vanishes as $C_i \to 0$. for the proposed estimator $\lamp$, we see that for sufficiently large $t_i$, $\fbt - a < 0$ so that there is no need for an adjusted predictor. 

\begin{rem}
    When $\fbt - a > 0 $, the improved predictor $\lambda_{i+}^{(P)} := \max (0,\lamp)$ is biased, therefore it is possible to further reduce the $\MSPE$ by bias correction. Yet, the unknown probability depends on unknown parameters $\beta$ and $a$ and also on unobserved covariates $X_i$ so may cause difficulties. However, this probability is often negligible, and in \Cref{secSIM} we show that under reasonable settings, there is often no need to be concerned about producing negative values.
\end{rem}

\section{Empirical predictors}\label{secEST}
Since $a$ and $\beta$ are generally unknown, estimators proposed in the previous section often need to be replaced by their empirical counterparts in practice. We will mention two sets of existing estimation schemes in this section, and readers can refer to the existing literature for further description. Both methods will be applied in \Cref{secSIM} to compare numerical results. 

We first estimate for $\beta$ with ordinary least squares using the observed covariates $\hat X_i$ instead of $X_i$, and denote the estimator as \[\hat\beta = \hat\beta^{(PR)} = (\hat X^T \hat X)^{-1}\hat X^T Z.\] Using $\hat\beta$ we then follow Equation (3.11) of Prasad and Rao (1990) \cite{PR1990} for estimation $a$, again replacing any instances of $X_i$ by $\hat X_i$, denoting the estimator as $\hat a^{(PR)}$. By doing so, we obtain
\begin{equation}\label{aPR}
   \hat a^{(PR)} = \frac{1}{m - p}\left[\sum_{i=1}^m (Z_i - \hat X_i\hat\beta)^2  - \frac{1}{4}\sum_{i=1}^m\{1 - \hat X_i (\hat X^T \hat X)^{-1}\hat X_i\}.\right]
\end{equation}
Next, we refer to Ybarra and Lohr (2008) \cite{YL2008} for another estimation scheme. To begin with, an estimator of $\beta$, which we denote $\hat\beta^{(YL)}$ can be obtained by solving the modified least squares equation:
\begin{equation}\label{ModLS}
    \left\{\sum_{i=1}^{m}w_i(\hat X_i^T \hat X_i - \frac{1}{t_i}\Sigma_i)\right\}\tilde\beta = \sum_{i=1}^{m}w_i \hat X_i Z_i
\end{equation}
for weights $w_1,...,w_m > 0$. When an inverse exists, we obtain 
\[\hat\beta^{(YL)} = \left(\sum_{i=1}^{m}w_i(\hat X_i^T \hat X_i - \frac{1}{t_i}\Sigma_i)\right)^{-1}\sum_{i=1}^{m}w_i \hat X_i Z_i.\]
Otherwise, we may utilize the Moore-Penrose inverse to obtain an estimator of $\beta$. For more details, we refer to Ybarra and Lohr (2008) \cite{YL2008}) for details. In our setting, we choose $w_1 = w_2 = ... = w_m$ for $i = 1,...,m$ so that \Cref{ModLS} reduces to 
\begin{equation}\label{ModLS2}
    \left\{\sum_{i=1}^{m}\hat X_i^T \hat X_i - \frac{1}{t_i}\Sigma_i \right\}\tilde\beta = \sum_{i=1}^{m} \hat X_i Z_i.
\end{equation}
This is because the sampling error variance is assumed to be constant $( = 1/4)$ across all areas $i = 1,...,m$, so it is much more straightforward to use equal weights across all areas. Moreover, doing so also avoids the need for multiple iterations, in which convergence is not guaranteed. We note that this is different from Ybarra and Lohr (2008) \cite{YL2008}, where the method requires iterating and updating weights.
We then follow their paper and estimate $a$ by
\begin{equation}\label{aYL}
\hat a^{(YL)} = \frac{1}{m - p}\sum_{i=1}^{m}\left\{(Z_i- \hat X_i \hat\beta^{(YL)})^2 - \hat\beta^{(YL)T} \Sigma_i \hat\beta^{(YL)} - \frac{1}{4}\right\}
\end{equation}
plugging in the estimator $\hat\beta^{(YL)}$ for $\beta$. 
\begin{rem}
In both estimation methods, estimates of $a$ can be negative, and we estimate $a$ as $0$ in such cases. 
\end{rem}

Upon obtaining estimates of $\beta$ and $a$, we may return to estimators in \Cref{secmain} and replace $\beta$ and $a$ for their estimated counterparts when appropriate. For instance, the weight $1 - \gam_i$ will be replaced by $1 - \dfrac{1/4}{1/4 +\hat a^{(PR)} + \frac{1}{t_i}\hat\beta^T\Sigma_i\hat\beta}$ when we use the first estimation scheme. In the case of $w_{0i}$, for the estimator to be applied without knowledge of $X_i$, $a$, and $\beta$, we may replace $X_i$ by $\hat X_i$, $a$ and $\beta$ by their estimates in \Cref{dMSE} and numerically obtain a solution. We attach a superscript of $ ^{(PR)}$ and $ ^{(YL)}$ to the empirical versions of the estimators when the estimates of $\beta$ and $a$ are derived using the methods from Prasad and Rao (1990) \cite{PR1990} and Ybarra and Lohr (2008) \cite{YL2008}, respectively.

\section{Simulation Study}\label{secSIM}
\subsection{Model Setup}
We simulate with true parameters $a = 0.2$ and $\beta = (1,0.5,1.5,1,0.3,2)$ and generate $\Sigma = \Sigma_i$ for $i = 1,...,m$ as Poisson random variables of mean 10 divided by 10. We compare results from $m = 20$ and $m = 100$, and in each case split the areas in half so that $t_i = 10$ for $i = 1,...,m/2$ and $t_i = 100$ for $i = m/2+1 ,...,m$. Entries of the design matrix $X$ are generated as independent $N(4,1)$ with a column of 1's as intercept. Apart from the intercept, the rows $X_i$ of $X$ are augmented with independent measurement errors $\Delta_i \sim N(0,\frac{1}{t_i}\Sigma)$, which are denoted as $\hat X_i$ like the rest of the paper. We also generate $v_i \sim N(0,a)$ and $e_i \sim N(0,\quarter)$ as independent Normal random variables respectively. We perform 10000 independent trials and compare performances using empirical $\MSPE$ (eMSPE).

We use the best predictor $\lambp$ (known $X$) and the direct predictor $\lamd$ as benchmarks, and compare $\lamp$, $\lamb$ (obtained by directly substituting $\hat X_i\beta$ into the best predictor of \Cref{BP}) and $\lamroot$ (with the optimal weight derived from the polynomial in \Cref{opt}). Again we note that knowledge of $X_i\beta$ is required to find $\lamroot$, despite it taking the considered form. Furthermore, we compare the empirical predictors $\lamppr$, $\lampyl$, $\lambpr$, $\lambyl$, $\lamrootpr$, and $\lamrootyl$ to inspect the effect of estimating the parameters and evaluate the performances of the two estimation schemes. The function \textit{polyroot} in the programming language R was used to obtain numerical roots for $\lamroot$, $\lamrootpr$, and $\lamrootyl$.

\begin{rem}
No negative values were produced in the numerical experiment so there was no need for adjusting any predictors to 0. We also numerically calculated the probabilities based on \Cref{probeq} and obtained numerical 0's in R.
\end{rem}

\subsection{Predictors with known model parameters}\label{NumRes}
We begin by comparing the numerical results of the predictors discussed in \Cref{secmain}. \Cref{tab1} compares the Theoretical and Numerical values of the Bias and $\MSPE$ obtained in two areas ($1^{st}$ and $20^{th}$) with different values of $t_i$. Note that theoretical values and empirical values are extremely close, showing the reliability of this numerical experiment. 

From \Cref{tab1}, most of the predictors are unbiased, and in the case of $\lamb$, the bias becomes negligible for large values of $t_i$. Referring to \Cref{thmMSE}, the $\MSPE$s of $\lamp$, $\lamb$, $\lamroot$ also halved when $t_i$ increases from 10 to 100 while $X_1\beta = 19.749$ and $X_{20}\beta = 19.514$ stayed roughly the same. We also note that when $t_i = 10$, the performance of $\lamb$ is worse than $\lamd$, meaning a naive substitution for $X_i \beta$ can backfire and produce worse estimates than the simplest ones.

The theoretical MSPEs of $\lamroot$ are marginally lower than $\lamp$, as justified in \Cref{secmain}, and on two occasions empirically $\lamp$ has even outperformed $\lamroot$ ($m = 20$, $t_i =100$ and $m=100$, $t_i = 100$) slightly. This justifies that $\lamp$ is an exceptionally good substitute for the optimal $\lamroot$, and this is shown further when inspecting the values of $w_i$, showing that $1 - \gam_i$ and $w_{0i}$ are extremely close, as shown in \Cref{rootbound}.

\subsection{Empirical Estimators}
\Cref{tab2} compares the performances of the empirical predictors. It is unsurprising that empirical predictors are inferior compared to those with known model parameters, regardless of estimation method, predictor, and the size of $t_i$. 

By inspecting the columns in \Cref{tab2} with the empirical predictors, we are able to see that the scheme of Prasad and Rao (1990) \cite{PR1990} outperforms that of Ybarra and Lohr (2008) \cite{YL2008} in all scenarios except the case of $\lamp$ for $t_i = 100$. Yet, it is still encouraging to see that regardless of the estimation scheme, both $\lamppr$ and $\lampyl$ still provide a superior alternative to the direct estimator $\lamd$. However, in the cases of $\lamb$ and $\lamroot$, there were cases when the empirical MSPE of the predictors with estimated coefficients perform significantly worse, especially for $\lamrootyl$ when $t_i = 100$. 

The method of Prasad and Rao (1990) along with the least squares estimator seems to be extremely stable overall, and in general, the empirical predictors did not perform much worse than when the model parameters are known. Comparing all predictors, it is apparent that empirical versions of our proposed estimator, $\lamppr$ and $\lampyl$ are stable and outperform the other predictors discussed in this paper. 

\begin{landscape}
    \begin{table}
    \centering
    \begin{tabular}{c|c|c|c|c|c|c}
         &&$\lambp$&$\lamd$ & $\lamp$ & $\lamb$&$\lamroot$\\
         \hline
         \multirow{5}{*}{$i = 1$, $t_{20} = 10$}&Theoretical Bias&0&0 & 0 & 0.2036& 0\\
         &Empirical Bias & 0.1047& 0.0955 &0.0872& 0.2762& 0.0872\\
         &Theoretical MSPE & 173.4032& 391.5527& 302.7520&491.2081 &302.7518\\
         &Empirical MSPE&174.4916 &388.7930 &302.8152& 490.6658& 302.8155\\
         & $w_i$ & 0.4444444& 1& 0.7746453& 0.4444444& 0.7747139\\ 
         \hline
        \multirow{5}{*}{$i = 20$, $t_1 = 100$}&Theoretical Bias&0&0 & 0 & 0.0203& 0 \\
         &Empirical Bias &-0.1989& -0.1907 &-0.1834& -0.1619& -0.1834\\
         &Theoretical MSPE & 169.3127 &382.3491 &196.4571& 200.4236 &196.4571\\
         &Empirical MSPE&169.5219 &385.2741 &196.1435 &199.5052& 196.1432 \\
         & $w_i$ & 0.4444444& 1 &0.5154432& 0.4444444 & 0.5154757\\
         \hline
    \end{tabular}
    \caption{Comparison of Theoretical and Empirical values of the Bias and $\MSPE$ for the predictors discussed in \Cref{secmain}. Specifically, the areas $i = 1$ and $i = m$ were inspected in the case $m = 20$. Note that $w_i$'s of $\lambp$ and $\lamb$ do not imply that $\lambp$ nor $\lamb$ are in our suggested form, it merely means $B = 0.4444444$.}
    \label{tab1}
\end{table}

    \begin{table}
    \centering
    \begin{tabular}{c|c|c|c|c|c|c|c|c|c|c|c|c|c}
         &&&$\lambp$&$\lamd$&$\lamp$&$\lamppr$&$\lampyl$&$\lamb$&$\lambpr$&$\lambyl$&$\lamroot$&$\lamrootpr$&$\lamrootyl$\\
         \hline
         \multirow{4}{*}{$m = 20$}&\multirow{2}{*}{$t_i = 10$}&Theoretical&222.7403 & 502.5611 & 388.7517 &&&  630.9554 &&&388.7516&&\\ 
         
         &&Empirical&223.4538 & 503.1421 & 389.3859  & 425.8178  & 430.4678 & 628.5475  & 443.4179  & 987.9840 &  389.3855 &  426.5739 &  592.9517\\

         \cline{2-14}
         
         &\multirow{2}{*}{$t_i=100$}&Theoretical&         252.9221 &  570.4703 & 293.4290 &&& 299.3134 &&& 293.4290&&\\
         
         &&Empirical&253.1840   &  567.6188   &  294.2121   &  419.9035   &  514.5984    & 300.4101   &  420.1814  &   680.1155    & 294.2122    & 421.3665   & 1303.7896\\
 
         \hline
        \multirow{4}{*}{$m = 100$}&\multirow{2}{*}{$t_i = 10$}&Theoretical&211.2762 &476.7671 &363.9533 &&&570.6043 &&&363.9532&&\\ 
         
         &&Empirical&211.7579   & 476.4815  &  364.3382  & 371.9718  & 371.8920 &  570.7392 &  379.6682  & 712.3375 &  364.3381 &  372.3204  & 504.7771\\

          \cline{2-14}

        &\multirow{2}{*}{$t_i=100$}&Theoretical& 240.2762 & 542.0170  & 276.3251 &&& 281.1800 &&& 276.3251&&\\
        
        &&Empirical&240.0288  & 539.7177 &  276.1711  & 332.6079 &  326.0317  & 281.1716  & 326.0345  & 367.7846  & 276.1712  & 334.7721  & 804.2483\\
        \hline
    \end{tabular}
    \caption{Empirical MSPEs of the eleven predictors discussed in either \Cref{secmain} or \Cref{secEST} of this paper. The values are averaged over all $m/2$ areas for each value of $t_i$. Predictors with estimated model parameters are shown to perform worse than their counterparts when model parameters are known, and the performance of the two estimation schemes discussed in \Cref{secEST} are compared. Note that since the $\MSPE$s depend on $X_i \beta$, one should the predictors within each row instead of comparing the same predictor across rows.}
    \label{tab2}
\end{table}
\end{landscape}

\section{Discussion}\label{secDIS}
We considered a simple and user-friendly class of predictors under the small area estimation model with measurement error and the square-root transform. We showed analytically what the best predictor among the class is and proposed an intuitive substitute, which was backed by theoretical and numerical evidence. We compared estimation schemes for empirical predictors and showed that in general naive substitutions can backfire if misused. 

Future avenues of research include $\MSPE$ estimation for the empirical predictors, where one can appeal to Bootstrap or Jackknife schemes. As discussed in \Cref{RangePres}, one may also consider bias adjustment for $\hat\lambda_{i+}$ in the case where the bias is positive.

\section{Acknowledgement}
The first author would like to thank the Joint Graduate School of Mathematics for Innovation (JGMI) of Kyushu University for their support. The second author’s research was partially supported by Kyushu University Diversity and Super Global Training Program for Female and Young Faculty (SENTAN Q) and JSPS KAKENHI grant number 22K01426.

\bibliography{KHG.manuscript}

\Addresses

\newpage
\section*{Appendix}\label{SecApp}
\subsection*{Proof of \Cref{thmMSE}}
Writing $C_i=-\{w_i^2(a+1/4) + (1-w_i)^2\fbt-a\}$ we have
\begin{align*}
    &\exp\left[\left\{\left(w_iZ_i + (1-w_i)\xhatb\right)^2+C_i-\lambda_i\right\}^2\right]\\
    =&\exp\left[\left\{w_iZ_i + (1-w_i)\xhatb\right\}^4+C_i^2+\lambda_i^2+2C_i\left\{w_iZ_i + (1-w_i)\xhatb\right\}^2-2C_i\lambda_i-2\lambda_i\left\{w_iZ_i + (1-w_i)\xhatb\right\}^2\right].
\end{align*}
Calculating each of the six terms separately, we have
\begin{align*}
    \exp[\left\{w_iZ_i + (1-w_i)\xhatb\right\}^4] &= \xib^4+6\xib^2\left\{w_i^2(a+\quarter)+(1-w_i)^2\fbt\right\}+3\left\{w_i^2(a+\quarter)+(1-w_i)^2\fbt\right\}^2\\
    \exp\left[C_i^2\right]&=C_i^2,\\
    \exp\left[\lambda_i^2\right] &= \xib^4 + 6\xib^2a+3a^2,\\
    \exp\left[2C_i\left\{w_iZ_i + (1-w_i)\xhatb\right\}^2\right] &= 2C_i\left[\xib^2 + \left\{w_i^2(a+\quarter)+(1-w_i)^2\fbt\right\}\right],\\
    \exp[-2C_i\lambda_i] &= -2C_i\left\{\xib^2+a\right\},
\end{align*}
and
\begin{align*}
    &\exp\left[-2\lambda_i\left\{w_iZ_i + (1-w_i)\xhatb \right\}^2\right]   \\ =&-2\exp\left[\theta_i^2\left\{w_i^2Z_i^2+2w_i(1-w_i)Z_i\xhatb+(1-w_i)^2\hxib^2\right\}\right]\\        
    =&-2\exp\left[\theta_i^2\left\{w_i^2(\theta_i+e_i)^2+2w_i(1-w_i)(\theta_i+e_i)\xhatb+(1-w_i)^2\hxib^2\right\}\right]\\    =&-2w_i^2\exp\left[\theta_i^4\right]-4w_i^2\exp\left[\theta_i^3\right]\exp\left[e_i\right]-2w_i^2\exp\left[\theta_i^2\right]\exp\left[e_i^2\right]-4w_i(1-w_i)\exp\left[\theta_i^3+\theta_i^3e_i\right]\xib
    -2(1-w_i)^2\exp\left[\theta_i^2\right]\exp\left[\hxib^2\right]\\
    =&-2\bigg[w_i^2\left\{\xib^4+6\xib^2+3a^2\right\}+w_i^2\quarter\left\{\xib^2+a\right\}\\
    &+2w_i(1-w_i)\left\{\xib^4+3\xib^2a\right\}+(1-w_i)^2\left\{\xib^2+a\right\}\left\{\xib^2+\fbt\right\}\bigg]\\
    =&-2\bigg[\xib^4+\xib^2\left\{6aw_i^2+\quarter w_i^2+6w_i(1-w_i)a+(1-w_i)^2(a+\fbt)\right\}
    +\left\{3a^2w_i^2+\quarter aw_i^2 +(1-w_i)^2a\fbt\right\}\bigg].
\end{align*}
Summing up and rearranging gives
\begin{align*}
     \exp\left[\left\{\hat\lambda_i(w_i) - \lambda_i \right\}^2\right] &= \xib^2\left[(4a+4\fbt+1)w_i^2-(8a+8\fbt)w_i+4a+4\fbt\right]\\
     &+ \bigg[3\left(w_i^2(a+\quarter)+(1-w_i)^2\fbt\right)^2 - \left(w_i^2(a+\quarter)+(1-w_i)^2\fbt-a\right)^2\\
     &+ 3a^2 -6a^2w_i^2-\frac{aw_i^2}{2}-2a\fbt(1-w_i)^2\bigg]
\end{align*}
as in \Cref{MSEwi}.

\end{document}